\documentclass{sig-alternate-10pt}
\usepackage{graphicx}     

\usepackage[utf8]{inputenc}
\usepackage[T1]{fontenc}
\usepackage{cite}

\usepackage{flushend}

\usepackage{url}

\newtheorem{prop}{Proposition}
\newtheorem{definition}{Definition}
\newtheorem{lemma}{Lemma}

\usepackage{times}
\usepackage{appendix}
\usepackage{xspace}
\newcommand{\crc}{CNC\xspace}

\begin{document}
\sloppy
\title{Clustered Network Coding for Maintenance\\ in Practical Storage Systems}
\author{Anne-Marie Kermarrec\\INRIA \and Erwan Le Merrer\\Technicolor \and Gilles Straub\\Technicolor \and Alexandre van Kempen\\Technicolor} 

\date{}
\maketitle
\thispagestyle{empty}
  
\begin{abstract}
  Classical erasure codes, \textit{e.g.} Reed-Solomon codes, have been
  acknowledged as an efficient alternative to plain replication to
  reduce the storage overhead in reliable distributed storage systems.
  Yet, such codes experience high overhead during the maintenance process.
  In this paper we propose a novel erasure-coded framework especially tailored for networked
  storage systems. Our approach relies on the use of random codes
  coupled with a clustered placement strategy, enabling the
  maintenance of a failed machine at the granularity of multiple
  files. Our repair protocol leverages network coding techniques to
  reduce by half the amount of data transferred during maintenance, as
  several files can be repaired simultaneously. This approach, as
  formally proven and demonstrated by our evaluation on a public
  experimental testbed, enables to dramatically decrease the bandwidth
  overhead during the maintenance process, as well as the time to repair a failure.
  In addition, the implementation is made as simple as possible, aiming at a deployment into practical systems. 

\end{abstract}

\section{Introduction}
\label{introduction}

Redundancy is key to provide a reliable service in practical systems
composed of unreliable components. Typically distributed storage
systems heavily rely on redundancy to mask ineluctable disk/node
unavailabilities and failures. While three-way replication (often
called triplication) is the standard means to obtain reliability with
redundancy, it is now acknowledged that \textit{erasure codes} can
dramatically improve the storage
efficiency~\cite{Weatherspoon2002}. In other words, for a given
reliability guarantee, the storage overhead for such codes is reduced
by order of magnitude compared to replication. Several major cloud
systems as those of Microsoft~\cite{azure} or Google~\cite{google}
have recently adopted erasure codes (more specifically
\textit{Reed-Solomon codes}). Facebook was experimenting them in 2010~\cite{facebook}.
There is thus a tangible move of cloud
operators from replication to erasure coding, allowing a more
efficient use of scalability-critical resources.


Reed-Solomon codes are the de facto standard of code-based redundancy
in practice. Yet, those codes have been designed and optimized to deal
with lossy communication channels, rather than specifically targeting
networked storage systems. In fact those codes only provide tolerance
to transient failures, the level of redundancy irrevocably decreasing
with host-node failures over time. An additional maintenance mechanism
is thus key to preserve the reliability of stored information over
time, as far as it is well known that storage systems have grown to a
scale where failures have become the norm.  However, Reed-Solomon
codes are precisely known to suffer from important overhead in terms
of bandwidth utilization and decoding operations when maintenance has
to be triggered. 
In order to address these two drawbacks,
architectural solutions have been proposed~\cite{hybrid}, as well as
new code designs~\cite{hierarchical,pyramid,Khan2012}, paving the way
for better tradeoffs between storage, reliability and maintenance
efficiency. The optimal tradeoff has been very recently provided by
Dimakis \& al~\cite{DimakisInfocomm} with the use of
\textit{network coding}.
However open issues regarding the feasibility of deploying those new
codes in practical distributed storage systems remain. Indeed, very
few studies evaluate how hard it is to implement theses
codes in a production system~\cite{Duminuco2009}, as most of them are
theoretical. Moreover those new codes are examined under the
simplifying assumption that only one file is stored per failed
machine, thus ignoring practical issues when dealing with the
maintenance of multiple files.

Interestingly enough, an appealing alternative for performance is to
use randomness.  Randomness can provide a simple and efficient way to
construct optimal codes \textit{w.h.p.}, as are Reed-Solomon ones,
while offering suitable properties in terms of maintenance.  Random
Codes have been identified as good candidates to provide fault
tolerance in distributed storage
systems~\cite{Dimakis2006,Gkantsidis05networkcoding,5743599,Acedanski05howgood}.
Yet, maintaining such promising codes has not been considered in
practice so far.  In this paper we propose a novel approach to
redundancy management, combining both random codes and network coding,
to provide an efficient maintenance protocol usable in
practice. \textbf{The main intuition behind our approach  
is to apply random codes and
  network coding at the granularity of clusters hosting, enabling to
  factorize the repair cost across several files at the same time.}
This mechanism is made as simple as possible, both in terms of design
and implementation with the purpose of leveraging the power of erasure
codes, while reducing its known drawbacks.

More specifically,  our contributions are the following:

\begin{enumerate}

\item We propose a novel maintenance mechanism 
  which combines a clustered placement strategy, random codes and
  network coding techniques at the node level (\emph{i.e.}, between
  different files hosted by a single machine). This approach is called
  \textit{\crc} in the sequel, for Clustered Network Coding. \crc
  enables to halve the data transferred compared to standard erasure
  codes during the maintenance process. The overhead in terms of
  decoding operations is also reduced by order of magnitude compared
  to the reparation process of classical erasure codes. Moreover, \crc
  enables \textit{reintegration} (\emph{i.e.}, the capability to
  reintegrate nodes which have been wrongfully declared as
  failed). Finally the network load is evenly balanced between nodes
  during the maintenance process, using a simple random selection.
  This enables the storage system to scale with the number of files to
  repair, as the available bandwidth is consumed as efficiently as it
  could be. Performance claims of \crc are formally proven.
  
\item We deployed \crc on a public execution platform, namely
  Grid5000\footnote{\url{http://www.grid5000.fr}}, to evaluate its
  benefits. In the typical setup of a datacenter storage system,
  data transferred, storage needs and repair time
  have been monitored.
  We compared our solution to both triplication and
  Reed-Solomon codes. Experimental results show that the data
  transferred for maintenance is reduced by half compared to codes
  while consuming the same storage space and providing the same data
  availability. The combination of the data transfer reduction, 
decoding operations avoidance, together with a clever use of the
  available bandwidth, has a strong impact on the efficiency of
  maintenance operation: the time to repair a failed node is
  dramatically reduced thus enhancing the whole reliability of the
  system.
 
\end{enumerate}
The rest of the paper is organized as follows. We first review
the background on maintenance techniques using erasure codes in
Section~\ref{background}. Our novel approach  is presented in Section~\ref{framework} and analyzed in Section~\ref{analysis} . We then
evaluate and compare it against state of the art approaches in
Section~\ref{evaluation}. Finally, we present related work in Section~\ref{related} and conclude this paper.

\section{Motivation and Background}
\label{background}
\subsection{Maintenance in Storage Systems}

Distributed storage systems are designed to provide reliable storage
service over unreliable
components~\cite{past,357007,502054,gfs}. One of the main challenges of
such systems is their ability to overcome unavoidable component failures~\cite{google,Vishwanath}.  Fault tolerance usually relies on data redundancy; the
classical triplication is the storage policy adopted by Hadoop~\cite{hadoop} or
the Google file system~\cite{gfs} for example. Data redundancy must be
complemented with a maintenance mechanism able to recover from the
loss of data when failures occur in order to preserve the reliability
guarantees of the system over time. Maintenance has already lain at
the very heart of numerous storage systems
design~\cite{voelker,stoica,1251200,1364689}. Similarly,
reintegration, which is the capability to reintegrate replicas stored
on a node wrongfully declared as failed, was shown
in~\cite{carbonite:nsdi06} to be one of key techniques to reduce the
maintenance cost.  
All these studies focused on the maintenance of replicas. While plain
replication is easy to implement and easy to maintain, it suffers from
a high storage overhead, typically $x$ instances of the same file are
needed to tolerate $x-1$ simultaneous failures. This high overhead is
a growing concern especially as the scale of storage systems keeps
increasing. This motivates system designers to consider erasure codes
as an alternative to replication.  Yet, using erasure codes
significantly increases the complexity of the system and challenges
designers for efficient maintenance algorithms.

\subsection{Erasure Codes in Storage Systems}

Erasure codes have been widely acknowledged as much more efficient
than replication~\cite{Weatherspoon2002} with respect to storage
overhead. More specifically, Maximum Distance Separable (MDS) codes
are optimal: for a given storage overhead (\textit{i.e.} the rate between the
original quantity of data to store and the quantity of data including
redundancy), MDS codes provide the optimal efficiency in terms of data
availability.  Let us now remind the reader about the basics of an MDS
code $(n,k)$: a file to store is split into $k$ chunks, encoded into
$n$ blocks with the property that any subset of $k$ out of $n$ blocks
suffices to reconstruct the file. Thus, to reconstruct a file of
$\mathcal{M}$ Bytes one needs to download \textit{exactly}
$\mathcal{M}$ Bytes, which corresponds to the same amount of data as if
plain replication were used. Reed-Solomon codes are a classical
example of MDS codes, and are already deployed in cloud-based storage
systems~\cite{azure,google}.  However, as pointed out
in~\cite{hybrid}, one of the major concern of erasure codes lies in
the maintenance process, which incurs an important overhead in terms
of bandwidth utilization as well as in decoding operations as
explained below.

\begin{figure}[!t] 
 \begin{center}    
 \includegraphics[scale=0.5]{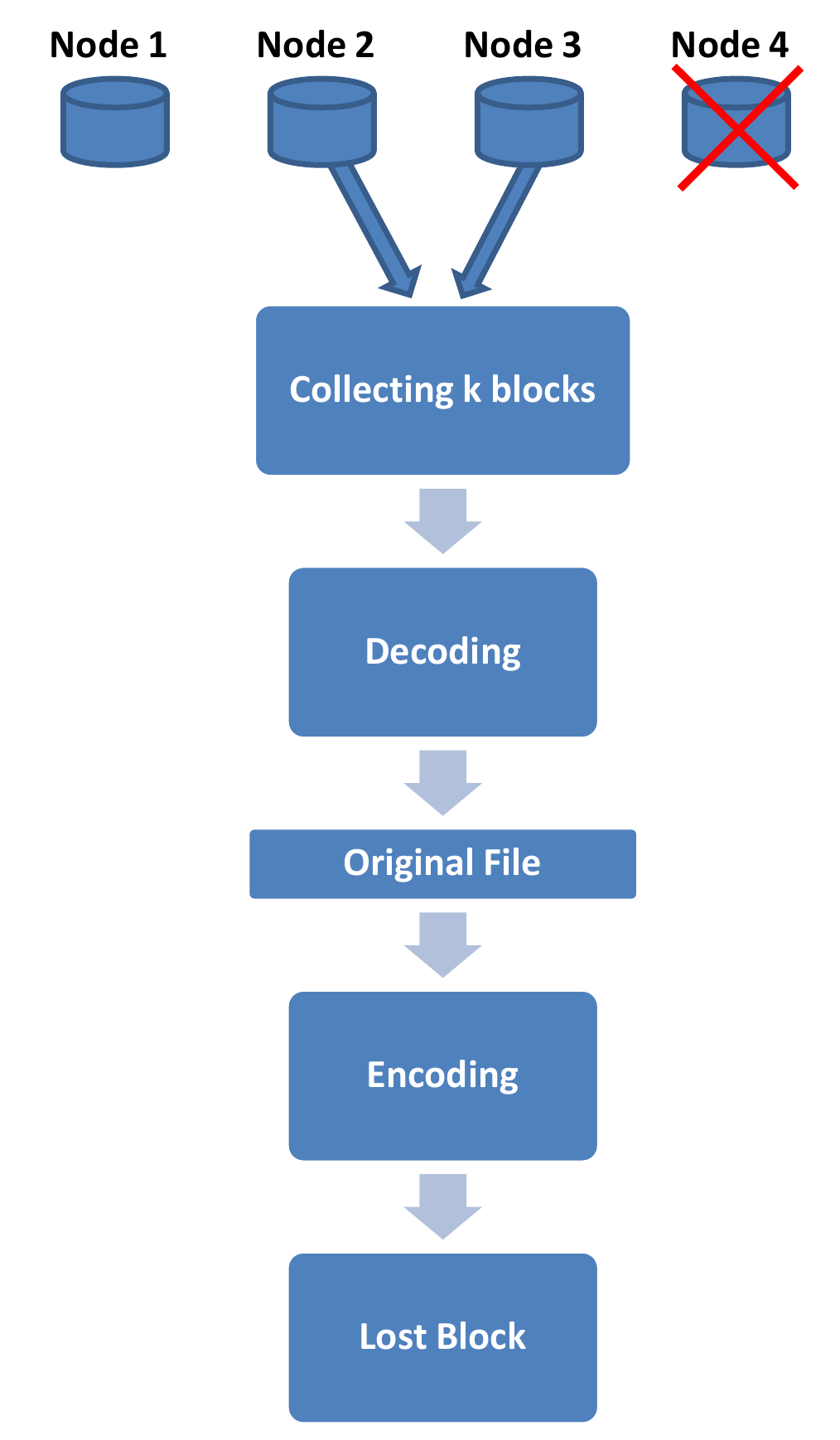} 
    \caption{Classical repair process of a single block. (k=2,n=4)}   
    \label{comparisonRepairProcess}   
  \end{center}     
   \end{figure}

\begin{figure}[t] 
 \begin{center}    
 \includegraphics[scale=0.34]{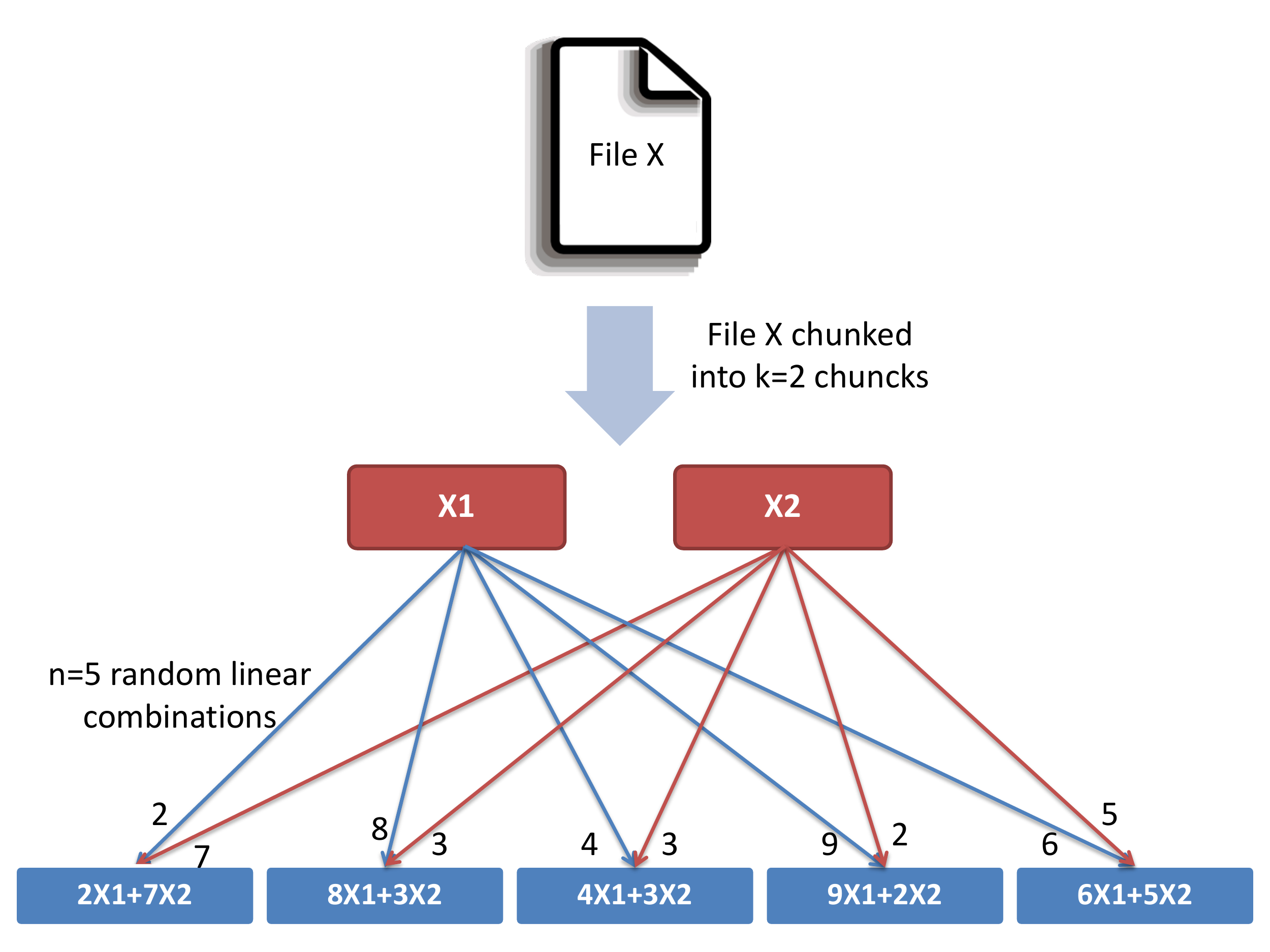} 
    \caption{Example of the creation process of encoded blocks using a random code. Any k=2 blocks is enough to reconstruct the file X}   
    \label{randomCode}   
  \end{center}     
   \end{figure}

\paragraph{Maintenance of Erasure Codes}
\label{classicalMaintenance}

When a node is declared as failed, all blocks of the files it was hosting need to be re-created on a \textit{new} node; we call this operation a \textit{repair} in the sequel. 
The repair process works as follows (see Figure~\ref{comparisonRepairProcess}): to repair one block of a given file, the new node first needs to download $k$ blocks of this file (\emph{i.e.}, corresponding to the size of the file) to be able to decode it. Once decoded, the new node can re-encode the file and then regenerate the lost redundant block. This must be iterated for all the lost blocks.
Three issues arise:
\begin{enumerate}
\item Repairing one block (typically a small part of a file) requires
  the downloading of enough blocks by the new node (\textit{i.e.} $k$) to reconstruct the entire
  file, and this must be done for all the blocks previously stored on
  the failed node.
\item  The new node must then decode the file, though it does not want to access it. Decoding operations are known to be time consuming especially for large files. 
\item Reintegrating a node which has been wrongfully declared as faulty is almost useless. This is due to the fact that the new blocks created during the repair operation have to be strictly identical to the lost ones for this is necessary to sustain the coding strategy~\footnote{This can be achieved either by a tracker maintaining the global information about all blocks or by the new node inferring the exact structure of the lost blocks from all existing ones.}. Therefore reintegrating a node results is having two identical copies of the involved blocks (the reintegrated ones and the new ones). Such blocks can only be useful if either the reintegrated node or the new node fails but not in the event of any other node failure.
\end{enumerate}

In order to mitigate these drawbacks, various solutions have been suggested.
Lazy repairs for instance as described in~\cite{1251200} consists in deliberately delaying the repairs, waiting for a successive amount of defects before repairing all the failures together. This enables to repair multiple failures while only suffering from bandwidth (\textit{i.e.} data transferred) and decoding overhead once. However delaying repairs leaves the system more vulnerable in case of a burst of failures.
Architectural solutions have also been proposed, as for example the \textit{Hybrid strategy}~\cite{hybrid}. This consists in maintaining one full replica stored on a single node in addition to multiple encoded blocks. This extra replica is thus utilized when repairs have to be triggered. However maintaining an extra replica on a single node significantly complicates the design, while incurring scalability issues.
Finally, new classes of codes have been designed~\cite{hierarchical,pyramid} which trade optimality in order to offer a better tradeoff between storage, reliability and maintenance efficiency.

\paragraph{A Case for Random Codes} 

In this paper, we argue that \textit{random linear codes} (random codes for short) may offer an appealing alternative to classical erasure codes in terms of storage efficiency and reliability, while considerably simplifying their maintenance process. Random codes have been initially evaluated in the context of distributed storage systems in~\cite{Acedanski05howgood}. Authors showed that random codes can provide an efficient fault tolerant mechanism with the property that no synchronization  between nodes is required. 
Instead, the way blocks are generated on each node is achieved independently in such a way that it will fit the coding strategy with high probability. Avoiding such synchronization is crucial in distributed settings, as also demonstrated in~\cite{Gkantsidis05networkcoding}. 


The basic principle of encoding a file using random codes is simple: each file is divided into $k$ chunks and the blocks stored for reliability are created as random linear combinations of these $k$ blocks (see Figure~\ref{randomCode}). All blocks, along with their associated coefficients, are then stored on $n$
 different nodes. Note that the additional storage space required for the coefficients is typically  negligible  compared to the size of each block.

   In order to reconstruct a file initially encoded with a given $k$,
   one needs to download $k$ different blocks of this file. Theory on
   random matrix over finite field ensures that if one takes $k$
   random vectors of the same subspace, these $k$ vectors are linearly
   independent with a probability which can be made arbitrary close to
   one, depending on the field size~\cite{Acedanski05howgood}. 
   This is a key difference with erasure codes, avoid any
   synchronization between nodes.  In other words, an encoded file can
   be reconstructed as soon as any set of $k$ encoded blocks is collected,
   and as already mentioned, this is optimal (MDS codes).

\section{Clustered Network Coding}
\label{framework}


Our \crc system is designed to sustain a predefined level
of reliability, \textit{i.e.} of data redundancy, by recovering from
failures with a limited impact on performances. We assume that the
failure detection is performed by a monitoring system, the description of which is out of the scope of this paper. We also assume that this system 
triggers the repair process, assigning new nodes to replace the faulty ones,
in charge of recovering the lost data and store it.

The predefined reliability level  is set by the storage system 
operator. This reliability level then directly translates into the
redundancy factor to be applied to files to be stored, with parameters
$k$ (number of blocks sufficient to retrieve a file) and $n$ (total
number of redundant blocks for a file). A typical scenario for using
\crc is a storage cluster like in the Google File
System~\cite{gfs}, where files are streamed into extents of the same
size, for example 1GB as in Windows Azure Storage~\cite{azure}.
These extents are erasure coded in order to save storage space.

\subsection{A Cluster-based Approach}

To provide an efficient maintenance, \crc relies on \textit{(i)} hosting  all blocks related to a set of files on a single cluster of nodes, and \textit{(ii)} repairing multiple files simultaneously .  This is achieved by combining the use of random codes, network coding and a cluster-based placement strategy. This enables to repair several files simultaneously, without requiring computationally intensive decoding operations, thus factorizing the costs of repair across the several multiple files stored by the faulty node. To this end, the system is  partitioned into disjoint clusters of $n$ nodes,
 so that each node of the storage system belongs to one and only one cluster. 
Each file to be stored is encoded using random codes and is associated to a single cluster. All blocks of a given file are then stored on the $n$ nodes of the same cluster. In other words, \crc placement strategy consists in storing blocks of two different files belonging to the same cluster on the same set of nodes, as illustrated on Figure~\ref{clusterPlacement}.

In such a setup, the storage system manager (\textit{e.g.} the master
node in the Google File System~\cite{gfs}) only needs to maintain two
data structures: an index which maps each file to one cluster and an
index by cluster which contains the set of the identifier of nodes in
this cluster. This simple data placement scheme leads to 
 significant data transfer gains and better load balancing, by clustering operations
on encoded blocks, as explained in the remaining part of this section.

\begin{figure}[t]    
\vspace{-0.8cm}
  \begin{center}    
 \includegraphics[scale=0.34]{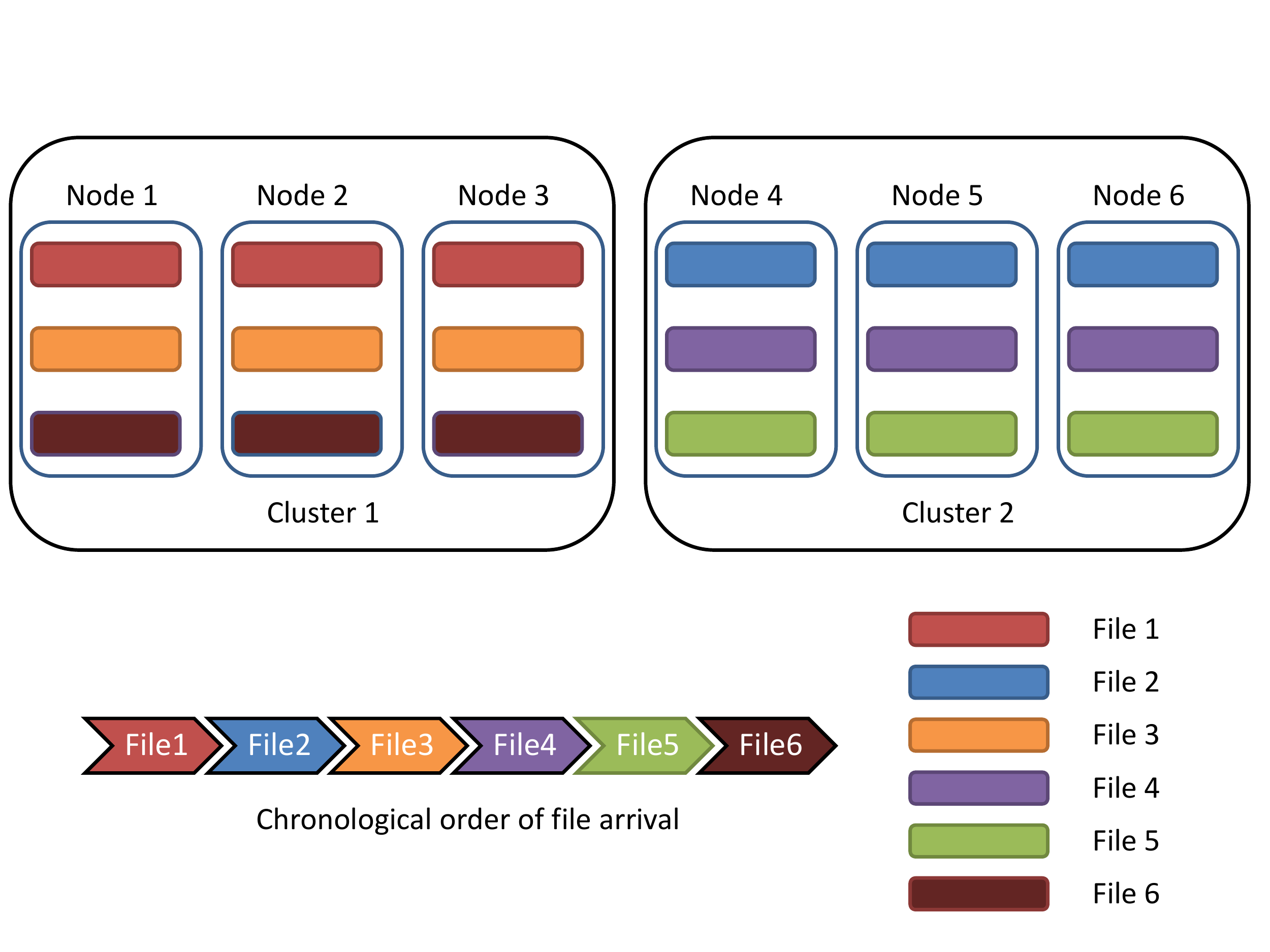} 
  \caption{Clustered placement for a $n=3$ redundant system}  
   \label{clusterPlacement}    
\end{center}     
   \end{figure}

\subsection{Maintenance of \crc}

When a node failure is declared, the maintenance operation must ensure that all the blocks hosted on the faulty node are repaired in order to preserve the redundancy factor and hence the predefined reliability level of the system. Repair is usually performed at the granularity of a file. Yet, a node failure typically leads to the loss of several blocks, involving several files. This is precisely this characteristic that \crc leverages. Typically, when a node fails,  multiple repairs are triggered, one for each particular block of one file that the failed node was storing.  Traditional approaches using erasure codes actually consider a failed node as the failure of all of its blocks.   \textbf{Instead, the novelty of \crc is to leverage network coding at the node level, \textit{i.e.} between different files on a particular cluster.} This is possible since \crc placement strategy clusters files so that all nodes of a cluster store the same files.  This technical shift enables to significantly reduce the data to be transferred during the maintenance process.

\begin{figure*}[t!] 
 \begin{center}    
 \includegraphics[scale=0.5]{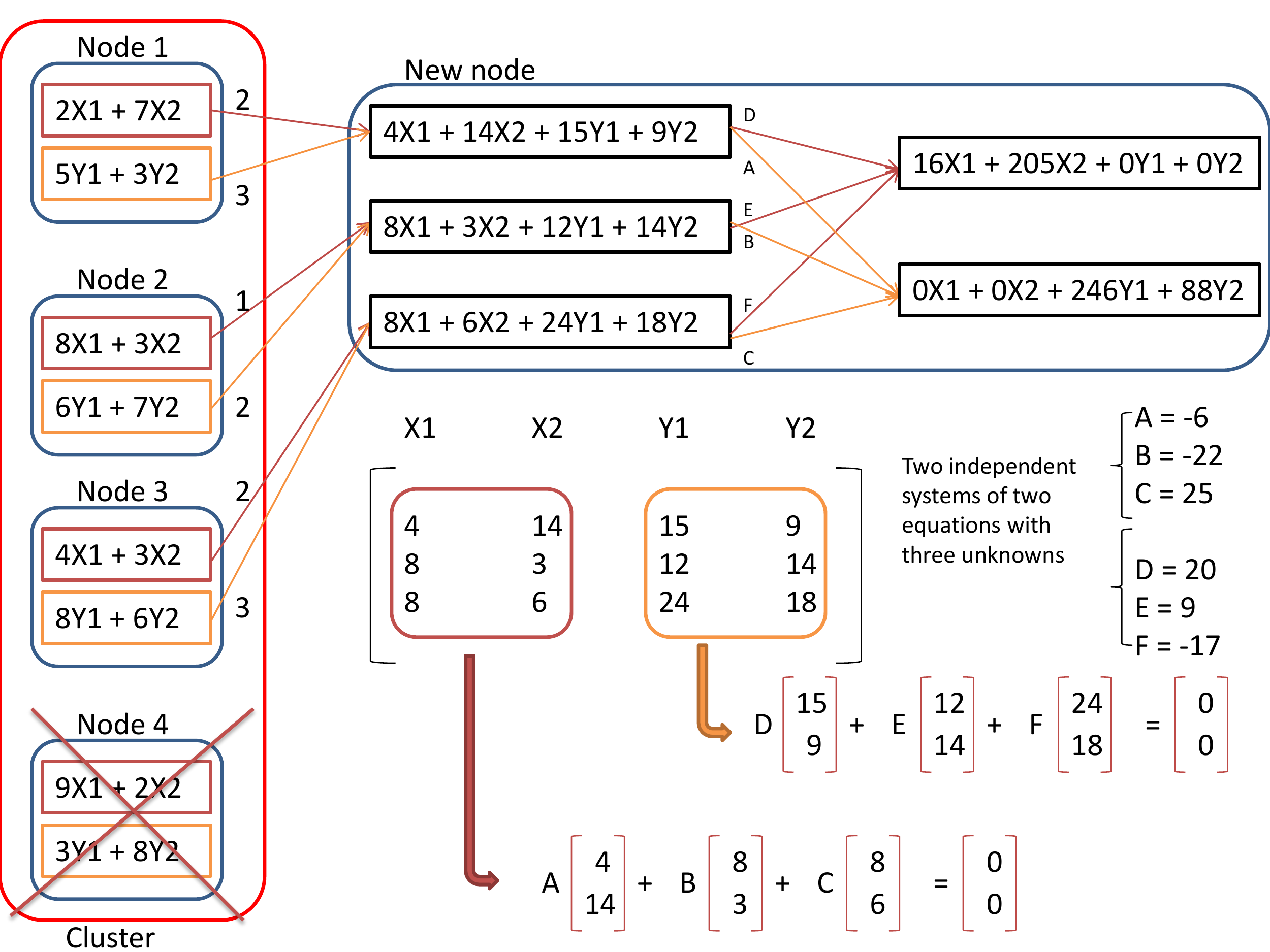} 
    \caption{Example of a \crc repair process, for the repair of a new node in a cluster of $4$ (with $k=2, n=4$).}   
    \label{repairMechanism}   
  \end{center}     
   \end{figure*}
   
\subsection{An Illustrating Example}
To provide the intuition of \crc, and before generalizing in the next
section, we now describe a simple example (see
Figure~\ref{repairMechanism}) involving  two files and a 4 node
cluster.  We consider two files X and Y of size $\mathcal{M}=1024$ MB, encoded with
random codes ($k=2$, $n=4$), stored on the 4 nodes of the same cluster
(\textit{i.e.} Nodes 1 to 4). File X is chunked into $k=2$ chunks $X_1$, $X_2$ as well as file Y into chunks $Y_1$ and $Y_2$. Each node stores two encoded blocks, one related to file X and the other to file Y which are respectively a random linear combination of $\lbrace X_1,X_2 \rbrace$ and $\lbrace Y_1,Y_2 \rbrace$. Each block has a size of $\frac{\mathcal{M}}{k}=512$ MB, thus each
node stores a total of $2\times512=1024$ MB.  We now consider the
failure of Node 4.

In a classical repair process, the new node asks to $k=2$ nodes their block corresponding to file X and Y and thus downloads 4 blocks, for a total of $4\times512=2048$ MB. This enables the new node to  decode the two files independently, and then re-encode each file to regenerate one block for file X and one for file Y and store them.  

Instead, \crc  leverages the fact that the encoded blocks related to files X and Y are stored on the same node and restored on the same new node to encode the files together rather than independently during the repair process. More precisely, if the nodes are able to compute a linear combination of their encoded blocks, we can prove that if $k=2$, only 3 blocks are sufficient to perform the repair of the two files X and Y. Thus the transfer of only $3$ blocks incurs the download of $3\times512=1536$ MB, instead of the $2048$ MB needed with the classical repair process.  In addition, this repair can be processed without decoding any of the two files. In practice, the new node has to contact the three remaining nodes to perform the repair. Each of the three nodes sends the new node a random linear combination of its two blocks with the associated coefficients. Note that the two files are now mixed, \textit{i.e.} encoded together. However, we want to be able to access each file independently after the repair. \textbf{The challenge is thus to create two new random blocks, with the restrictions that one is only a random linear combination of the X blocks, and the other of the Y blocks.}
In this example, finding the appropriate coefficients in order to cancel the $X_i$ or $Y_i$,  comes down to solve two independent systems of two equations with three unknowns as shown in Figure~\ref{repairMechanism}. The intuition is that, as coefficients of these equations are random, these two systems are always solvable \textit{w.h.p.}. The new node then makes two different linear combinations of the three received blocks according to the previously computed coefficients, $(A,B,C)$ and $(D,E,G)$ in the example. Thereby it creates two new independent random blocks, one related to file X and one to file Y. The repair is then performed, saving the bandwidth consumed by the transfer of one block \emph{i.e.}, $512$ MB in this example. 

\subsection{\crc: The General Case}
\label{generalCase}

We now generalize the previous example for any $k$. We first define a \textit{RepairBlock} object: a RepairBlock is a random linear combination of two encoded blocks of two different files stored on a given node. RepairBlocks are transient objects which only exist during the maintenance process \emph{i.e.}, RepairBlocks are never stored permanently.

We are now able to formulate the core technical result of this paper; the
following proposition applies in a context where different files are
encoded using random codes with the same $k$, and the encoded blocks
are placed according to the cluster placement described in the
previous section.

\begin{prop} 
\label{prop1}
In order to repair two different files, downloading $k+1$ RepairBlocks
from $k+1$ different nodes is a sufficient condition.
\end{prop}

Repairing two files jointly actually comes down to create one new
random block for each of the two files; the formal proof of this
proposition is given in Appendix. This proposition implies that
instead of having to download $2k$ blocks as with Reed-Solomon codes
when repairing, \crc decreases that need to
only $k+1$.  Other implications and analysis are detailed in the next section.

We shall notice that the encoded blocks of the two files do not need to have the same size. In case of different sizes, the smallest is simply zero-padded during the network coding operations as usually done in this context; padding is then removed at the end of the repair process.

\begin{figure}[t] 
\vspace{-0.5cm}
 \begin{center}    
 \includegraphics[trim = 1cm 0cm 1cm 1cm,scale=0.34]{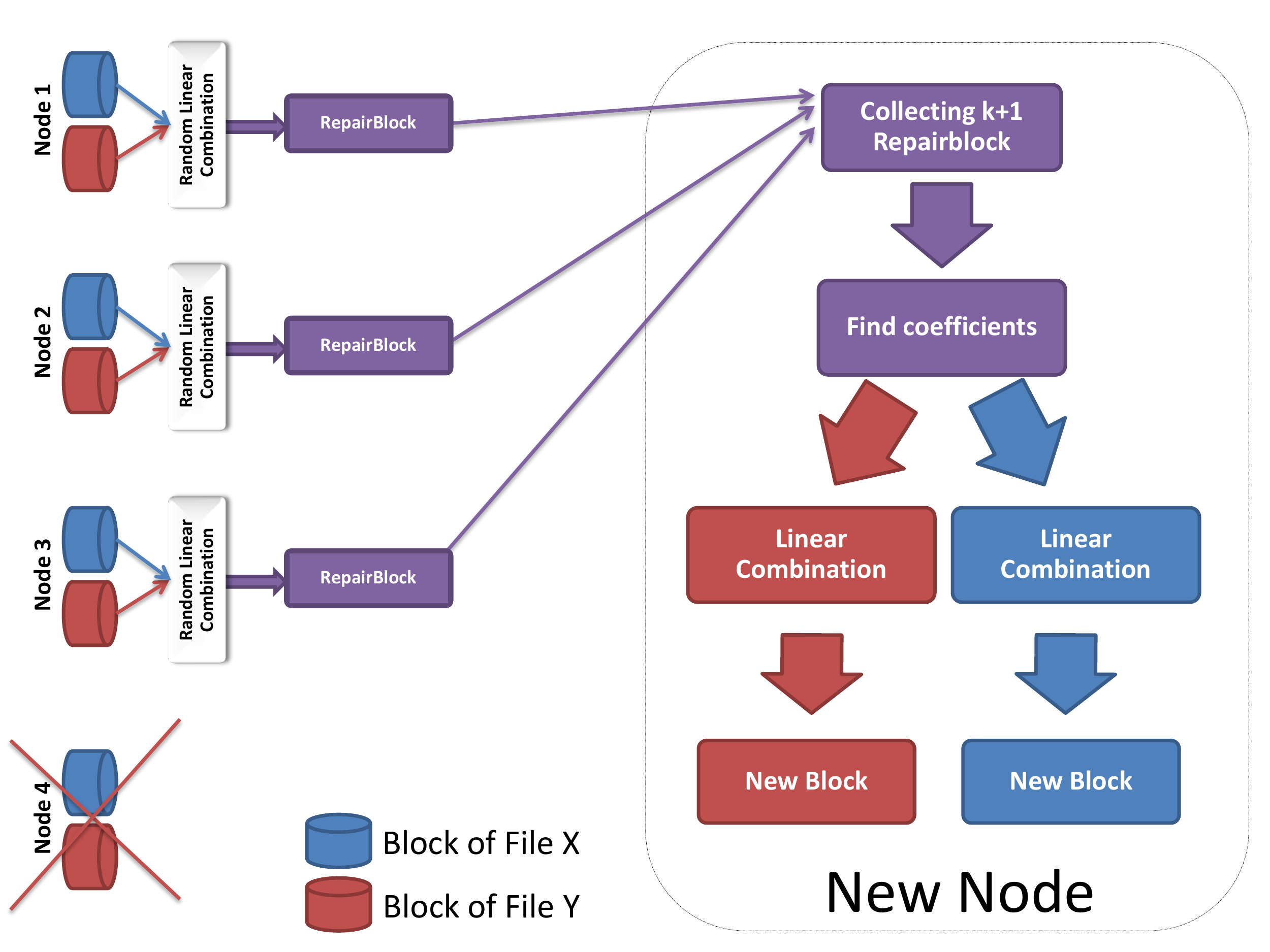} 
    \caption{One iteration of the repair process, at the end of which two encoded blocks are repaired.}   
    \label{ourRepairProcess}   
  \end{center}     
   \end{figure}

Figure~\ref{ourRepairProcess} describes one iteration of the process at the end of which two encoded blocks are repaired. Each of the $k+1$ nodes sends a RepairBlock to the new node, which then combines them to restore the two lost encoded blocks. However nodes usually store far more than two blocks, implying multiple iterations of the process described in Figure~\ref{ourRepairProcess}.
More formally, to restore a failed node which was storing $x$ blocks, the repair process must be iterated $\frac{x}{2}$ times. In fact, as two new blocks are repaired during each iteration, the number of iteration is halved compared to the classical repair process.
Note that in case of an odd number of blocks stored, the repair process is iterated until only one block remains. The last block is repaired downloading $k$ blocks of the corresponding file which are then randomly combined to conclude the repair. The overhead related to the repair of the last block in case of an odd block number vanishes with a growing number of blocks stored.

The fact that the repair process must be iterated several times can also be leveraged to balance the bandwidth load over all the nodes in the cluster.
Only $k+1$ nodes over the $n$ of the cluster are selected at each iteration of the repair process; as all nodes of the cluster have a symmetrical role, a different set of $k+1$ nodes can be selected at each iteration.
In order to leverage the whole available bandwidth of the cluster, \crc makes use of a random selection of these $k+1$ nodes at each iteration.
In other words, for each round of the repair process, the new node selects $k+1$ nodes over the $n$ cluster nodes randomly. 
Doing so, we show that every node is evenly loaded \emph{i.e.}, each node sends the same number of RepairBlocks in expectation. 

More formally, let N be the number of RepairBlocks sent by a given node. In a cluster where $n$ nodes participate in the maintenance operation, for $T$ iterations of the repair process, the average number of RepairBlocks sent by each node is : 
\begin{equation}
\label{meanLoad}
E(N) = T\frac{k+1}{n}
\end{equation} 

The proof is given in Appendix.
An example illustrating this proposition is provided in the next section.

\section{CNC Analysis}
\label{analysis}

\begin{figure}[t] 
 \begin{center}    
 \includegraphics[trim = 1cm 0cm 1cm 1cm,scale=0.83]{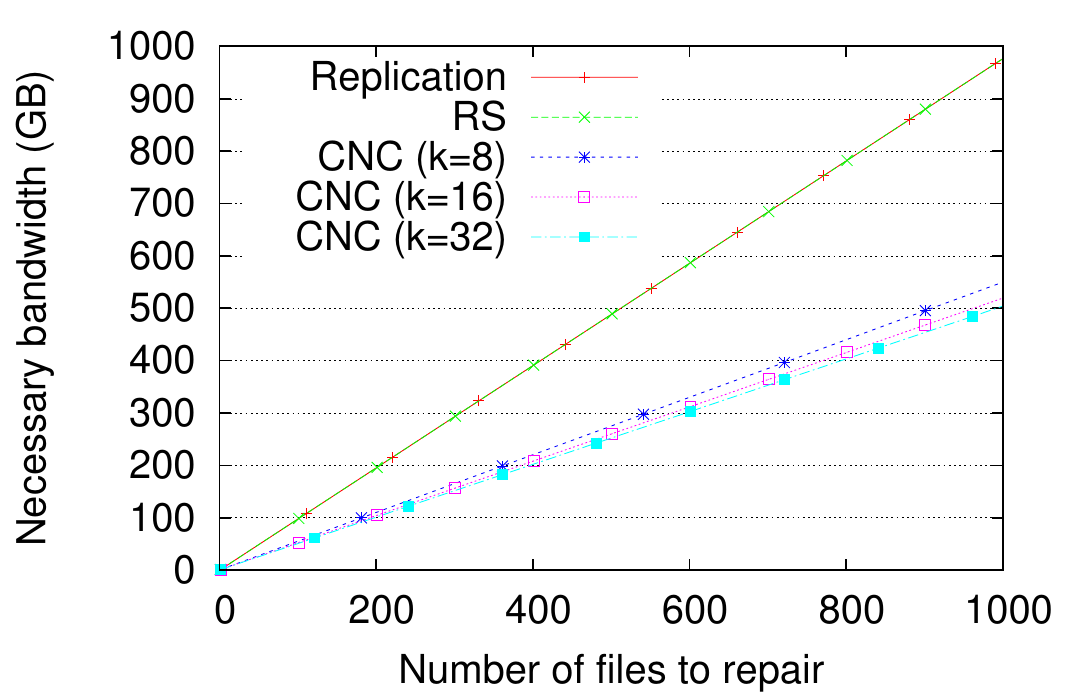} 
    \caption{Necessary amount of data to transfer to repair a failed node, 
according to the selected redundancy scheme (files of $1$ GB each).}   
    \label{necessaryBandwidth}   
  \end{center}     
   \end{figure}

The novel maintenance protocol proposed in the previous section enables \textit{(i)} to significantly reduce the amount of data transferred during the repair process; \textit{(ii)} to balance the load between the nodes of a cluster; \textit{(iii)} to avoid computationally intensive decoding operations and finally \textit{(iv)} to provide useful node reintegration. The benefits are detailed below.

\subsection{Transfer Savings}

A direct implication of Proposition~\ref{prop1} is that for large
enough values of $k$, the data to transfer required to perform a
repair is halved; this directly results in a better usage of available
bandwidth within the datacenter. To repair two files in a
classical repair process, the new node needs to download at least $2k$
blocks to be able to decode each of the two files. Then the ratio
$\frac{k+1}{2k}$ (\crc over Reed-Solomon) tends to $1/2$ as larger values of $k$ are used.

The exact necessary amount of data $\sigma(x,k,s)$ to repair $x$ blocks of size $s$ encoded with the same $k$ is given as follows: 

$$
\sigma(x,k,s) = \left\{
    \begin{array}{ll}
        \frac{x}{2}s(k+1) & \mbox{if x is even }\\
         \frac{x}{2}s(k+1+\frac{k-1}{x}) & \mbox{if x is odd}
    \end{array}
\right.
$$

An example of the transfer savings is given in Figure~\ref{necessaryBandwidth}, for $k=16$ and a file size of $1GB$.  


We described in \crc, through Proposition~\ref{prop1}, the need to
repair lost files by groups of two. One can wonder whether there is a
benefit in grouping more than two files during the repair.  In fact a
simple extension of Proposition~\ref{prop1} is that to group
$\mathcal{G}$ files together, a sufficient condition is that the new
comer downloads $(\mathcal{G}-1)k +1$ RepairBlocks from
$(\mathcal{G}-1)k +1$ distinct nodes over the $n$ ones in the cluster.
Firstly, this implies that the new node must be able to contact many
more nodes than $k+1$.  Secondly, we can easily see that the gains made
possible by \crc are maximal when grouping the file by two: savings in
data transfer when repairing are expressed by the ratio
$\frac{(\mathcal{G}-1)k +1 }{\mathcal{G}k}$. The minimal value of this
ratio ($\frac{1}{2}$, which is equivalent to the maximal gain) is
obtained for $\mathcal{G}=2$ and large value of $k$.

A second natural question is whether or not downloading fewer than $(\mathcal{G}-1)
k+1$ RepairBlocks to group $\mathcal{G}$ files together is possible. We can
positively answer this question, as the value $(\mathcal{G}-1)k +1 $ is only a
\textit{sufficient} condition. In fact, if nodes do not send random
combinations, but carefully choose the coefficients of the
combination, it is theoretically possible to download less
RepairBlocks. However, as $\mathcal{G}$ grows,  finding the "adequate" coefficients becomes computationally  
intractable, especially for large values of $k$. These coefficients can be found in some cases using 
\textit{interference alignment} techniques (see for example~\cite{DimakisSurvey}). 
However details of these techniques are outside the scope of this paper as no efficient algorithm is known to solve this problem to date. This then calls for the use of the simpler operation \emph{i.e.}, $\mathcal{G}=2$ as we have presented in this paper.

\subsection{Load Balancing}

\begin{figure}[!t]
 \begin{center}
 \includegraphics[scale=0.35]{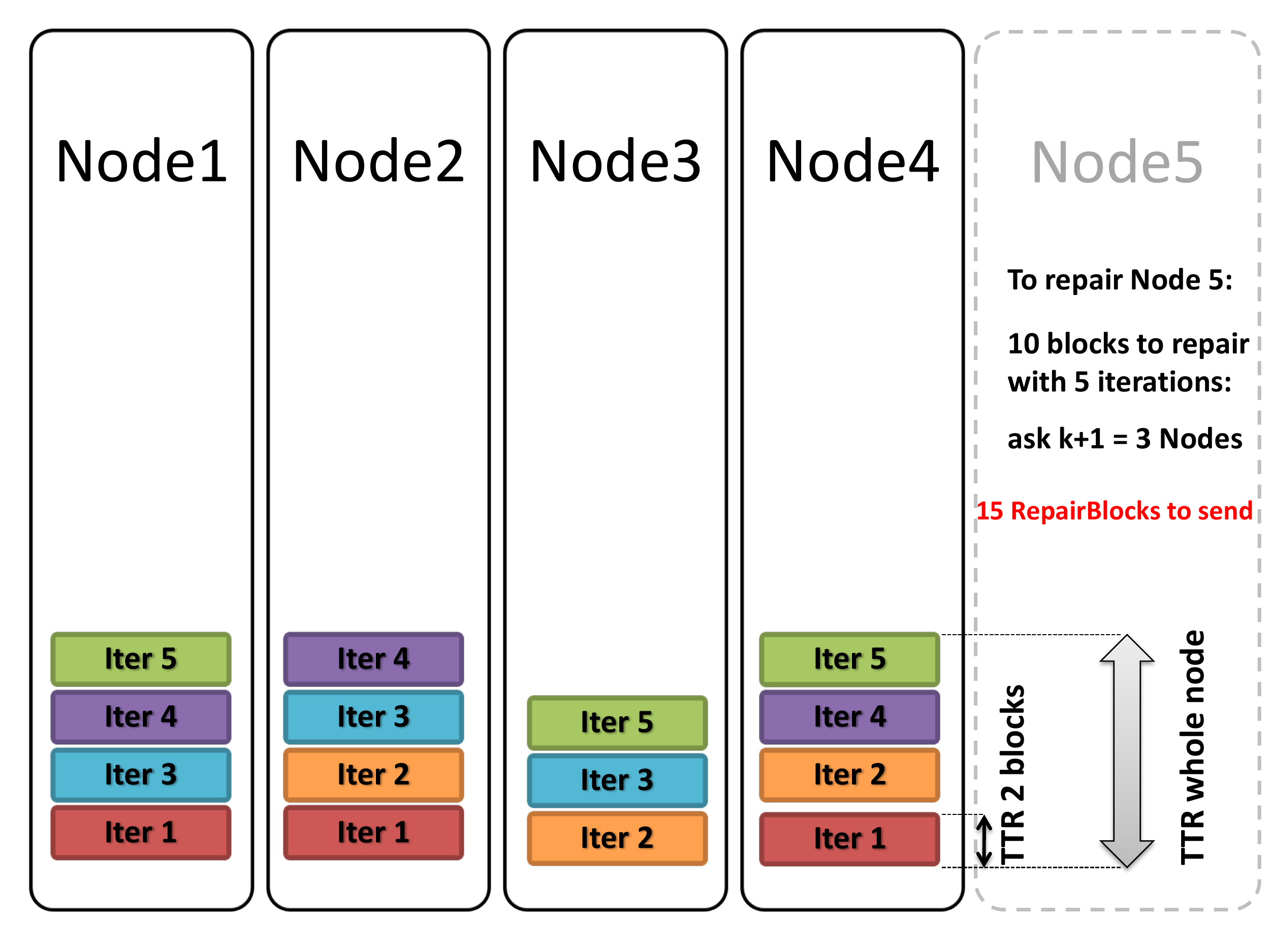}  \\  
 \caption{Natural load balancing for blocks queried when repairing a failed node (node $5$), for $10$ blocks to restore.}    
   \label{loadexample}       
 \end{center} 
  \end{figure}

As previously mentioned, when a node fails, the repair process  is iterated
as many  times  as needed to repair all lost blocks.  
\crc ensures that the load over remaining nodes is balanced during maintenance;
Figure~\ref{loadexample} illustrates this.
This example involves a $5$ node cluster, storing $10$ different
files encoded with random codes ($k=2$). Node $5$ has failed, involving the
loss of $10$ blocks of the $10$ files stored on that cluster.  
Nodes $1$ to $4$ are available for the repair process. 

\crc provides a load balanced approach, inherent to the random
selection of the $k+1=3$ nodes at each round.  In addition, only $T=5$
iterations of the repair process are necessary to recreate the $10$
new blocks, as each iteration enables to repair $2$ blocks at the same
time. The total number of RepairBlocks sent during the whole
maintenance is $T\times(k+1)=15$, whereas the classical repair process
needs to download $20$ encoded blocks.  The random selection ensures
in addition that the load is evenly balanced between the available
nodes of the cluster. Here, nodes $1$,$2$ and $4$ are selected during
the first repair round, then nodes $2$, $3$ and $4$ during the second
round and so forth.  The total number of RepairBlocks is balanced
between all available nodes, each sending
$\frac{T\times(k+1)}{n}=\frac{15}{4}=3.75$ RepairBlocks on average. As a
consequence of using the whole available bandwidth in parallel, contrary to sequentially fetching blocks for only a subset of nodes, the
Time To Repair (TTR) a failed node is also greatly reduced. This is confirmed 
 experimentally in Section~\ref{evaluation}.

\subsection{No Decoding Operations}

Decoding operations are known to be time consuming and should
therefore only be necessary in case of file accesses. 
While the use of
classical erasure codes requires such decoding to take place upon
repair, \crc avoids those cost-intensive operations.  In fact, no file
needs to be decoded at any time in \crc: repairing two blocks only
requires to compute two linear combinations instead of decoding the
two files. However the output of our repair process is strictly equivalent
if files had been decoded. 
This greatly simplify  the repair process 
over classical approaches.
As a consequence, the time to
perform a repair is reduced by order of magnitude compared to the
classical reparation process, especially when dealing with large files
as confirmed by our experiments (Section~\ref{evaluation}).

\subsection{Reintegration}

The decision to declare a node as failed is usually performed
using timeouts; this is typically a decision prone to
errors~\cite{carbonite:nsdi06}.  In fact, nodes can be wrongfully
timed-out and can reconnect once the repair is done.  While the
longer the timeouts, the fewer errors are made, adopting large
timeouts may jeopardize the reliability guarantees, typically in the
event of burst of failures.  The interest of reintegration is to
be  able to leverage  the fact that nodes which have been wrongfully
timed-out are reintegrated in the system. 
While this idea has already been explored using
replication~\cite{carbonite:nsdi06}, reintegration has not been
addressed when using erasure codes.

As previously mentioned, when using classical erasure codes, 
the repaired blocks have to be strictly identical
to the lost ones. Therefore reintegrating a failed node
in the system is almost useless for this results in two identical copies of the 
lost and repaired blocks.  Such blocks can only be useful  in the event of the failure of two specific nodes, the wrongfully timed-out node and the new node. 

On the contrary, reintegration is always useful when deploying \crc.
More precisely, every single new block can be leveraged to compensate
for the loss of any other block and therefore are useful in the
event of the failure of any node. Indeed, new created blocks are simply
new random blocks, thus different from the lost ones while being
functionally equivalent.  Therefore each new block contributes to the
redundancy factor of the cluster. Assume that a node which has been
wrongfully declared as failed returns into the system. A repair has
been performed to sustain the redundancy factor while it turned out not
to be necessary.  This only means that the system is now one repair
process ahead and  can  leverage this unnecessary repair to avoid 
triggering a new instance of the repair protocol when  the next failure occurs.

\section{Evaluation}
\label{evaluation}

In order to confirm the theoretical savings provided by the \crc
repair protocol, in terms of bandwidth utilization and decoding
operations, we deployed \crc over a public experimental
platform. We describe hereafter the implementation of the system and
\crc experimental results. 

\subsection{System Overview}
\label{implementation}
\begin{figure}[t!]    
  \begin{center}    
 \includegraphics[trim = 1cm 3cm 1cm 1cm,scale=0.34]{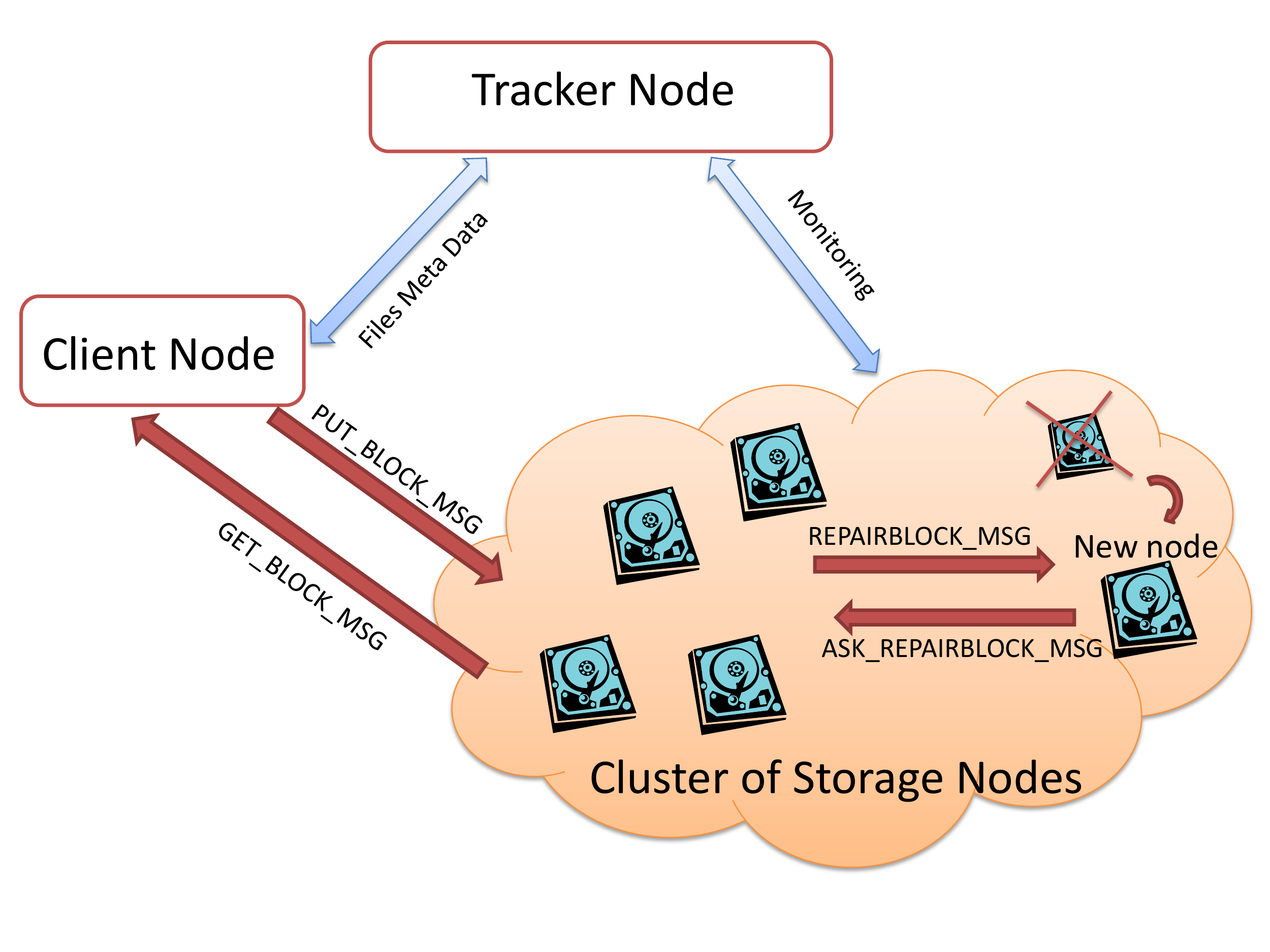} 
  \caption{System overview}  
   \label{archi}    
\end{center}     
   \end{figure}

   We implemented a simple storage cluster with an architecture
   similar to Hadoop~\cite{hadoop} or the Google File
   System~\cite{gfs}.  This architecture is composed of one
   \textit{tracker node} that manages the metadata of files, and
   several \textit{storage nodes} that store the data. This set of
   storage nodes forms a cluster as defined in
   Section~\ref{framework}. The overview of the system architecture is
   depicted in Figure~\ref{archi}. Client nodes can PUT/GET the data
   directly to the storage nodes, after having obtained their IP
   addresses from the tracker. In case of a storage node failure, the
   tracker initiates the repair process and schedules the repair jobs.

All files to be stored in the system are encoded using random codes with the same $k$.
Let $n$ be the number of storage nodes in the cluster, then $n$ encoded blocks are created for each file, one for each storage node.
Remind that the system can thus tolerate $n-k$ storage node failures before
files are lost for good.

\paragraph{PUT/GET and Maintenance Operations}

In the case of a PUT operation, the client first encodes blocks.
The coefficients of the linear combination associated to each encoded block are appended at the beginning of the block.
Those $n$ encoded blocks are sent to the $n$ storage nodes of the cluster using a {\small PUT\_BLOCK\_MSG}.
A {\small PUT\_BLOCK\_MSG} contains the encoded information, as well as the hash of the corresponding file.
Upon receipt of a {\small PUT\_BLOCK\_MSG}, the storage node stores the encoded block using the hash as filename.

To retrieve the file, the client sends a {\small GET\_BLOCK\_MSG} to at least $k$ out of the $n$ nodes of the cluster.
A {\small GET\_BLOCK\_MSG} only contains the hash of the file to be retrieved.
Upon receipt of a {\small GET\_BLOCK\_MSG} the storage node sends the block corresponding to the given hash.
As soon as the client has received $k$ blocks, the file can be recovered.


In case of a storage node failure, a new node is selected by the tracker to replace the failed one.
This new node sends a {\small ASK\_REPAIRBLOCK\_MSG} to $k+1$ storage nodes. 
An {\small ASK\_REPAIRBLOCK\_MSG} contains the two hashes of the two blocks which have to be combined following the repair protocol described in Section~\ref{framework}.
Upon receipt of an {\small ASK\_REPAIRBLOCK\_MSG}, the storage node  combines the two encoded blocks corresponding to the two hashes, and sends 
the resulting block back to the new node. As soon as $k+1$ blocks are received, the new node can regenerate two lost blocks. This process is iterated until all lost blocks are repaired.

\subsection{Deployment and Results}

We deployed the system previously described on the Grid5000 execution platform.
The experiment ran on $33$ nodes connected with a $1$GB network. Each node has 2Intel Xeon L5420 CPUs $2.5$ GHz, $32$GB RAM and a $320$GB hard drive.
We randomly chose $32$ storage nodes to form a cluster, as defined in Section~\ref{framework}. The last remaining node was elected as the tracker.
All files were encoded with $k=16$, and we assumed that the size of each inserted file is 1GB. This size is used in Windows Azure Storage for sealed extents which are erasure coded~\cite{azure}. 

\paragraph{Scenario}

In order to evaluate our maintenance protocol, we implemented a first phase of i files insertion in the cluster, and artificially triggered a repair during the second phase. 
According to the protocol previously described, the tracker selects a
new node to replace the faulty node, to which it sends to the list of
IP addresses of the storage nodes.  The new node then directly asks
RepairBlocks to storage nodes, without any intervention of the
tracker, until it recovers as many encoded blocks as the failed node
was storing.  We measured the time to repair a failed node depending
on the number of blocks it was hosting. The time to repair is defined
as the time between the reception of the list of IPs, and the time all
new encoded blocks are effectively stored on the new node.  We
compared \crc against a classical maintenance mechanism (called RS),
which would be used with Reed-Solomon codes as described in
Section~\ref{classicalMaintenance} and with standard replication.  All
the presented results are averaged on three independent experiments.
This small number of experiments can be explained by the fact that
  Grid5000 enables to make a reservation on a whole cluster of
nodes in isolation ensuring that experiments are highly reproducible
and we observed a standard deviation under 2 seconds for all values.

\paragraph{Coding} 

\begin{figure}[t] 
 \begin{center}    
 \includegraphics[trim = 1cm 0cm 1cm 1cm,scale=0.8]{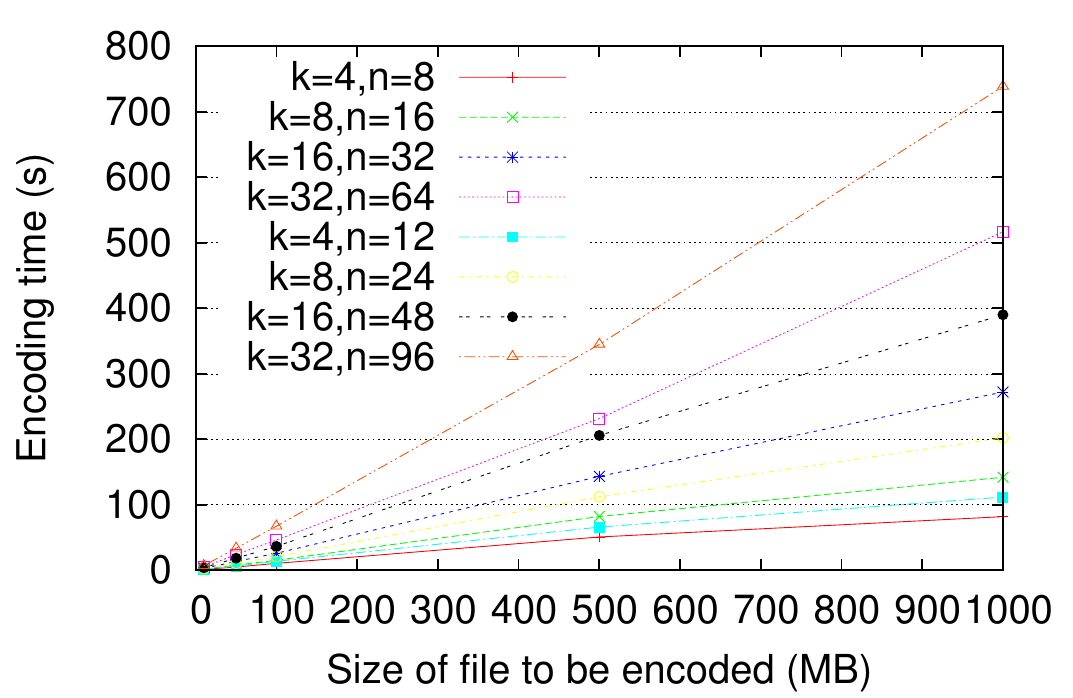} 
    \caption{Encoding time depending on file size when using random codes.}   
    \label{EncodingTime}   
  \end{center}     
   \end{figure}

We developed a Java library to deal with arithmetic operations over a
finite field\footnote{This library will be made public along with the
  paper.}. In this experiment,  arithmetic operations are
performed over a finite field with $2^{16}$ elements as it enables to
treat data as a stream of unsigned short integers (16 bits). Additions
and subtractions correspond to XOR operations between two
elements. Multiplications and divisions are performed in the logspace
using lookup tables which are computed offline.  This library enabled
us to implement classical matrix operations over finite fields, such
as linear combinations, encoding and decoding of files.

We measure the encoding time when using random codes for various code
rates, depending on the size of the file to be encoded. Results are
depicted on Figure~\ref{EncodingTime}. We show that for a given ($k,n$)
the encoding time is clearly linear with the file size. For example
with ($k=16$,$n=32$) the encoding time for a file of size $512$MB and $1$GB
are respectively $143$ and $272$ seconds.  In addition, the encoding time
increases with $k$ and with the code rate, as more encoded blocks have
to be created. For instance, a file of 1GB with $k=16$ is encoded in $272$
seconds for a code rate 1/2 ($n=32$), whereas $390$ seconds are
necessary for a code rate 1/3 ($n=48$).

\paragraph{Transfer Time}
\label{transferTime}

\begin{figure}[t] 
\vspace{-1.02cm}
 \begin{center}    
 \includegraphics[trim = 3cm 3cm 1cm 1cm,scale=0.35,keepaspectratio=true,height=5.8cm]{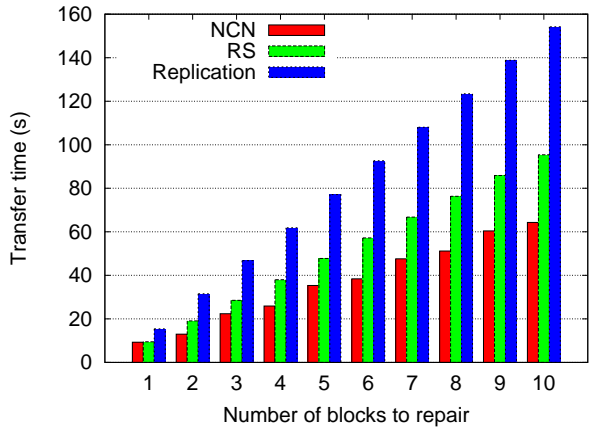}
    \caption{Time to transfer the necessary quantity of data to perform a complete repair}   
    \label{TotalTimeTo}   
  \end{center}     
   \end{figure}

We evaluated in this experiment the time to transfer the whole quantity of data needed to perform a complete repair for \crc, RS and replication, depending on the number of blocks to be repaired. In order to quantify the gains provided by \crc in isolation, we disabled the load balancing part of the protocol in this experiment. In other words, the same set of nodes is selected for all iterations of the repair process. The results are depicted on Figure~\ref{TotalTimeTo}. 

Firstly, we observe on the figure that \crc consistently outperforms
the two alternative mechanisms. As \crc incurs the transfer of a much
smaller amount of data, the time to transfer the blocks during the
repair process is greatly reduced compared to both RS and
replication. For instance, to download the necessary quantity of data
to repair a node which was hosting 10 blocks related to 10 different
files, \crc only requires $64$ seconds whereas RS and replication
requires respectively $95$ and $154$ seconds on average. It should
also be noted that no coding operations are done in this experiment,
except for \crc as nodes have to compute a random linear combination
of their encoded blocks to create a RepairBlock before sending
it. This time is taken into account, thus explaining why the transfer
time for \crc is not exactly halved compared to RS.

A second observation is that \crc also scales better with the number
of files to be repaired. As opposed to \crc, both RS and replication
involve transfer times for multiple files which are strictly
proportional to the time to transfer a single file. For example RS and
\crc requires $9$ seconds to download a single file, but RS requires
$95$ seconds to download $10$ files, while \crc only requires $64$
seconds for the same operation.

Finally, replication leads to the highest time to transfer. This is
mainly due to the fact that replication does not leverage parallel
downloads from several nodes as opposed to \crc and RS. 
Yet replication does not suffer from computational costs, which can
dramatically increase the whole repair time of a failure as shown in
the next section.

\paragraph{Repair Time}

\begin{figure}[t] 
 \begin{center}    
 \includegraphics[trim = 3cm 3cm 1cm 1cm,scale=0.35]{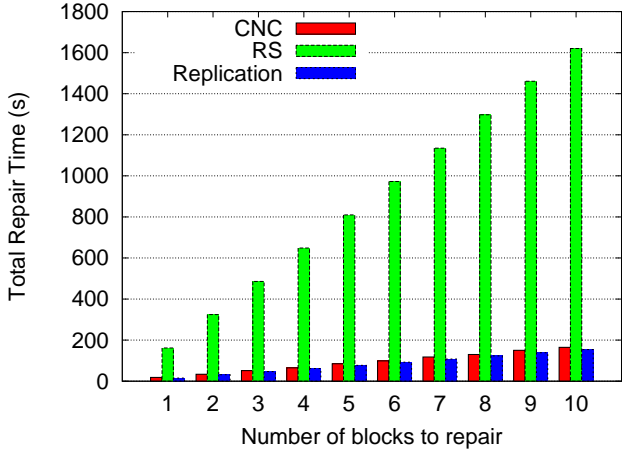} 
    \caption{Transfer time}   
    \label{TimeTo}   
  \end{center}     
   \end{figure}

In this experiment, we measured the whole repair time of a failure, depending on the number of blocks (related to different files) the failed node was storing. The results, depicted on Figure\ref{TimeTo}, include both the transfer times, evaluated in the previous section, as well as coding times. Thereby it represents the effective time between a failure is detected and the time it has been fully recovered. As replication does not incur any coding operations, the time to repair is simply the time to transfer the files. Note that for the sake of fairness, we enable the load balancing mechanism both for \crc and RS. 

Figure\ref{TimeTo} shows that the repair time is dramatically reduced when using \crc compared to RS, especially with an increasing  number of files to be repaired.  For instance to repair a node which was hosting $10$ blocks related to $10$ different files, \crc and replication require respectively $165$ and $154$ seconds while RS needs $1620$ seconds on average.

These time savings are mainly due to the fact that decoding operations
are avoided in \crc. In fact, the transfer time is almost negligible
compared to the computational time for RS.  The transfer time only 
represents  $6$\% of the time to repair a node which was hosting
$10$ blocks related to $10$ different files with RS. This clearly
emphasizes the interest of avoiding computationally intensive tasks
such as decoding during the maintenance process.

We can also observe that time to repair a failure with \crc is nearly
equivalent to the one when using replication.  As shown in
Figure~\ref{TotalTimeTo}, replication transfer times are much higher
than \crc ones, but this is counter-balanced by the fact that some
coding operations are necessary in \crc.  In other words, \crc saves
time compared to replication during the data transfer, but these
savings are cancelled out due to linear combination computations.
Finally our experiments show that, as opposed to RS, \crc scales as
well as replication with the number of files to be repaired.

\paragraph{Load Balancing} 

As shown in Section~\ref{generalCase}, \crc provides a natural load balancing feature.  The random selection of nodes from which to download blocks  during the maintenance process ensures that the load is evenly balanced between nodes. In this section, firstly  we experimentally verify that nodes are evenly loaded, then we evaluate the impact of this load balancing on the transfer time for both \crc and RS. 

Figure~\ref{load} shows the number of blocks sent by each of the 32 nodes of the cluster for a repair of a node which was storing $100$ blocks when using \crc. This involves $50$ iterations of the protocol, where at each iteration, $k+1=17$ distinct nodes send a RepairBlock. 
We observe that all nodes send a similar number of blocks \emph{i.e.}, nearly $26$, in expectation. This is consistent with the expected 
value analytically computed, according to Equation~\ref{meanLoad} as $\frac{25 \times 17}{32} = 26.5625$.

Figure~\ref{loadtime} depicts the transfer time for both RS and \crc
depending on the number of files to be repaired. We compare the
transfer time between the load balanced approach (\crc-LB and RS-LB),
and its counterpart which involves a fixed set of nodes, as done in
Section~\ref{transferTime}. Results show that transfer times are
reduced when load balancing is enabled, as the whole available
bandwidth can be leveraged. In addition, time savings due to the load
balance increases as more files have to repaired.

   \begin{figure}[!t] 
 \begin{center}    
 \includegraphics[trim = 3cm 3cm 1cm 1cm, scale=0.35]{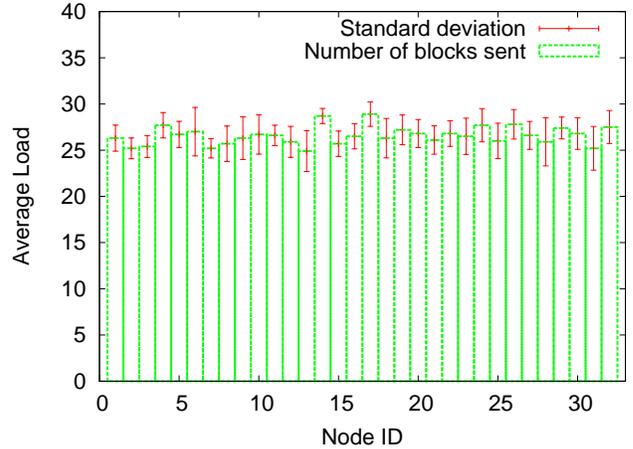} 
    \caption{Average load on each node in the cluster}   
    \label{load}   
  \end{center}     
   \end{figure}

\section{Related Work}

The problem of efficiently maintaining erasure-coded content has
triggered a novel research area both in theoretical and practical
communities. Design of novel codes tailored for networked storage
system has emerged, with different purposes.
	
For instance, in a context where partial recovering may be tolerated,
priority random linear codes have been proposed in~\cite{4711046} to
offer the property that critical data has a higher opportunity to
survive node failures than data of less importance.  Another point in
the code design space is provided by self-repairing codes~\cite{src}
which have been especially designed to minimize the number of nodes
contacted during a repair thus enabling faster and parallel
replenishment of lost redundancy.
	
In a context where bandwidth is the scarcest resource, network coding
has been shown to be a promising technique which can serve the
maintenance process.  Network coding was initially proposed to improve
the throughput utilization of a given network
topology~\cite{networkFlow}.  Introduced in distributed storage
systems in~\cite{DimakisInfocomm}, it has been shown that the use of
network coding techniques can dramatically reduce the maintenance
bandwidth. Authors of ~\cite{DimakisInfocomm} derived a class of
codes, namely \textit{regenerating codes} which achieve the optimal
tradeoffs between storage efficiency and repair bandwidth. In spite of
their attractive properties, \textit{regenerating codes} are mainly
studied in an information theory context and lack of practical
insights. Indeed, this seminal paper provides theoretical bounds on
the quantity of data to be transferred during a repair, without
supplying any explicit code constructions.  The computational cost of
a random linear implementation of these codes can be found
in~\cite{Duminuco2009}. A broad overview of the recent advances in
this research area are surveyed in~\cite{DimakisSurvey}.
     
Very recently, authors in~\cite{Khan2012}
and~\cite{Papailiopoulos2012} have designed new code especially
tailored for cloud systems. Paper~\cite{Khan2012} proposed a new class
of Reed-Solomon codes, namely \textit{rotated Reed-Solomon} codes with
the purpose of minimizing I/O for recovery and degraded read.
\textit{Simple Regenerating Codes}, introduced
in~\cite{Papailiopoulos2012}, trade storage efficiency to reduce the
maintenance bandwidth while providing exact repairs, and simple XOR
implementation.
       
Some other recent works ~\cite{Hu2012,5978919} aim to bring network
coding into practical systems. However they rely on code designs which
are not MDS, thus consuming more storage space, or are only able to
handle a single failure hence limiting their application context.

 \begin{figure}[t] 
 \begin{center}    
 \includegraphics[trim = 3cm 3cm 1cm 1cm, scale=0.35]{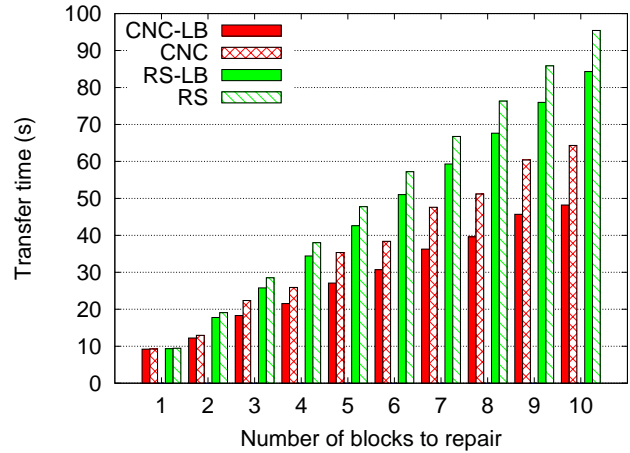} 
    \caption{Impact of load balancing on Time To Transfer}   
    \label{loadtime}   
  \end{center}     
   \end{figure}
\label{related}

\section{Conclusion}
\label{conclusion}
While erasure codes, typically Reed-Solomon, have been acknowledged as
a sound alternative to plain replication in the context of reliable
distributed storage systems, they suffer from high costs, both
bandwidth and computationally-wise, upon node repair. This is due to
the fact that for each lost block, it is necessary to download enough
blocks of the corresponding file and decode the entire file before
repairing.  

In this paper, we address this issue and provide a novel
code-based system providing high reliability and efficient maintenance
for practical distributed storage systems. The originality of our
approach, \crc, stems from a clever cluster-based placement strategy,
assigning a set of files to a specific cluster of nodes combined with
the use of random codes and network coding at the granularity of
several files.  \crc leverages network coding and the co-location of
blocks of several files to encode files together during the repair. This provides a
significant decrease of the bandwidth required during repair, avoids
file decoding and provides useful node reintegration.  We provide a
theoretical analysis of \crc. We also implemented \crc and deployed it
on a public testbed. Our evaluation shows dramatic improvement of \crc
with respect to bandwidth consumption and repair time over both plain
replication and Reed-Solomon-based approaches.

\section{Acknowledgment}
 
Experiments presented in this paper were carried out using the
Grid'5000 experimental testbed, being developed under the INRIA
ALADDIN development action with support from CNRS, RENATER and several
Universities as well as other funding bodies (see
https://www.grid5000.fr).  We thank Steve Jiekak for the
implementation of the Java library for finite field operations.

\bibliographystyle{abbrv} 
\bibliography{biblio}

\appendix


\subsection*{Proof of Proposition 1}

 
\begin{lemma}
\label{lemma1} 
 A linear combination of independent random variables chosen uniformly in a finite field $\mathbb{F}_q$ also follows a uniform distribution over $\mathbb{F}_q$.
\end{lemma}

\begin{proof}
Let $\mathbb{S}_N$ be the random variable defined by the linear combination of $N$ random variables $\lbrace X_1,X_2,...,X_N\rbrace$ . These $N$ random variables are independent and take their values uniformly in the finite field $\mathbb{F}_q$.

 \begin{equation*}
\label{linearcomb}
\mathbb{S}_N = \sum_{i=1}^{N} \alpha_i X_i   
 \end{equation*}

with $\forall i, X_i  \in \mathbb{F}_q$, and $\alpha_i \in  \mathbb{F}_q^\ast$

We show by recurrence that if $ \forall i, Pr(X_i = x_i) = \frac{1}{q}$ then  
$Pr(\mathbb{S}_N = s_N)  = \frac{1}{q}$

The case $N=1$ is trivial. 
Let first show that for $N=2$ the proposition is true.

\begin{equation*}
\begin{split}
\mathbb{S}_2 &= \alpha_1 X_1 + \alpha_2 X_2 \\
Pr(\mathbb{S}_2 = s_2) &= Pr( \alpha_1 X_1 + \alpha_2 X_2 = s_2)\\
& = \sum_{x_1=0}^{q-1} Pr(X_1 = x_1) Pr(X_2 = \frac{s_2 - \alpha_1x_1}{\alpha_2})\\
& = \sum_{x_1=0}^{q-1} \frac{1}{q}  \frac{1}{q} = q  \frac{1}{q}  \frac{1}{q}\\
& = \frac{1}{q}\\
\end{split}
\end{equation*}
The proposition is thus true for $N=2$.
We suppose that it is true for all $N$, and prove that it is true for $N+1$.

\begin{equation*}
\begin{split}
\mathbb{S}_{N+1}&=\mathbb{S}_N+\alpha_{N+1}X_{N+1}\\
Pr(\mathbb{S}_{N+1}=s_{N+1})&=Pr(\mathbb{S}_N+\alpha_{N+1}X_{N+1}= s_{N+1})\\
&=\sum_{x_{N+1}=0}^{q-1}[Pr(X_{N+1}=x_{N+1})\\
& \times Pr(\mathbb{S}_N=s_{N+1}-\alpha_{N+1} x_{N+1})]\\
&=\sum_{x_{N+1}=0}^{q-1}\frac{1}{q}\frac{1}{q}= q\frac{1}{q}\frac{1}{q}\\
&=\frac{1}{q}\\
\end{split}
\end{equation*}
\end{proof}

\begin{definition}
A random vector $\vec{V}$ in a vector space $ X = \operatorname{span} \lbrace X_1,X_2,...,X_k \rbrace$  where $X_i \in \mathbb{F}^l_q$ is defined as : 

 \begin{equation*}
\label{linearcomb2}
V = \sum_{i=1}^{k} \alpha_i X_i   
 \end{equation*}

where the $\alpha_i$ coefficients are chosen uniformly at random in the field $\mathbb{F}_q $, ie. $Pr(\alpha_i = \alpha) = \frac{1}{q} ,  \forall \alpha \in \mathbb{F}_q$.
\end{definition}

Let X be the vector space defined as $\operatorname{span} \lbrace X_1,X_2,...,X_k \rbrace$.
Let Y be the vector space defined as $\operatorname{span} \lbrace Y_1,Y_2,...,Y_k \rbrace$.
No assumptions are made on $X_i$ and $Y_i$ except that they are all in $\mathbb{F}^l_q$.
In fact as $X_i$ and $Y_i$ are file blocks, it is not possible to ensure linear independence for example.

Let $B^i_x$ be an encoded block of the file $F_x$ stored on node $i$. $B^i_x$ is a random linear combination of the ${\lbrace X_1,X_2,...,X_k \rbrace}$, thus $B^i_x \in \operatorname{span} \lbrace X_1,X_2,...,X_k \rbrace = X$ which is a subspace of $\mathbb{F}^l_q$ of dimension 
$\operatorname{Dim} (X) \leq k$.

\begin{lemma} 
$\forall \operatorname{Dim} (X)$, $B^i_x$ is a random vector in $X$.
\end{lemma}

\begin{proof}

Let $\mathcal{B}$ be the largest family of linearly independent vectors of $\lbrace X_1,X_2,...,X_k \rbrace$

$\forall l \mid X_l \not\in \mathcal{B}, \exists! \lbrace b^l_1,...,b^l_j  \rbrace$ such that
$X_l = \sum_{\tiny{j\mid X_j \in \mathcal{B}}} b^l_j X_j$

\begin{equation*}
\begin{split}
B^i_x &= \sum_{j=1}^{k} a^i_j X_j\\
& = \sum_{\tiny{j\mid X_j \in \mathcal{B}}} a^i_j X_j +       \sum_{\tiny{l\mid X_l \not\in \mathcal{B}}} a^i_l X_l\\
& = \sum_{\tiny{j\mid X_j \in \mathcal{B}}} a^i_j X_j +      
\sum_{\tiny{l\mid X_l \not\in \mathcal{B}}} a^i_l \sum_{\tiny{j\mid X_j \in \mathcal{B}}} b^l_j X_j \\
& = \sum_{\tiny{j\mid X_j \in \mathcal{B}}} (a^i_j + \sum_{\tiny{l\mid X_l \not\in \mathcal{B}}} a^i_lb^l_j) X_j \\
\end{split}
\end{equation*}

From Lemma~\ref{lemma1}, all the coefficients of the linear combination  are random over $\mathbb{F}_q$ thus $B^i_x$ is a random vector in 
$\operatorname{span} (\mathcal{B})$.

As $\operatorname{span} (\mathcal{B}) = \operatorname{span} \lbrace X_1,X_2,...,X_k \rbrace = X $

Then $B^i_x$ is a random vector in $X$.
\end{proof}

Let $D^i$ be the random linear combination of two stored blocks by the node $i$ with $i \in [1,k+1]$.

\begin{equation*}
\begin{split}
D^i &= \delta^i_x B^i_x + \delta^i_y B^i_y\\
& =  \delta^i_x (\sum_{j=1}^{k} a^i_j X_j) +  \delta^i_y (\sum_{l=1}^{k} a^i_l Y_l)\\
&  = D^i_x + D^i_y\\
\end{split}
\end{equation*}

By definition,  $D^i \in \operatorname{span} \lbrace X_1,X_2,...,X_k,Y_1,...,Y_k \rbrace$

$D^i_x = \sum_{j=1}^{k}\delta^i_x a^i_j X_j$

As $\delta^i_x$ are chosen randomly in $\mathbb{F}_q$, then from Lemma~\ref{lemma1}, $D^i_x$ is a random vector in X.

 Let's take a family $\lbrace D^1_x,...,D^{k+1}_x\rbrace$.\\
 As $\operatorname{Dim} (X)  \leq k $ it exists $\lbrace \alpha_1,
 ...,\alpha_{k+1} \rbrace \neq 0$ such that $\sum_{i=1}^{k+1}\alpha_iD^i_x = 0$
 
 Thus :
 
\begin{equation*}
\begin{split}
\sum_{i=1}^{k+1}\alpha_iD^i &= \sum_{i=1}^{k+1}\alpha_iD^i_x +
\sum_{i=1}^{k+1}\alpha_iD^i_y\\
& = \sum_{i=1}^{k+1}\alpha_iD^i_y\\
\end{split}
\end{equation*}

As $\alpha_i$ are chosen independently with $D^i_y$ then 
new vector is a random vector in $Y$.
The reasoning is identical to get the new vector in $X$, thus completing the proof.

\subsection*{Proof of Equation (1)}

During the repair process, the load on each node can be evaluated using a Balls-in-Bins model.
Balls correspond to a block to be downloaded while bins represents the nodes which are storing the blocks.
For each iteration of the repair protocol, k different nodes are selected to send a repair block.
This corresponds to throwing k identical balls into n bins, with the constraints that once a bin has received a ball, it can not receive another ball at this round. In other words exactly k different bins are chosen at each round.
\begin{lemma} 
\label{lemma3}
At each round $i$, the probability that a given bin has received one ball is $\frac{k}{n}$
\end{lemma}
\begin{proof}
Let A be the event "the bin contains one ball at round $i$". 
Thus $\overline{A}$ corresponds to the event "the bin is empty at round $i$". 
$Pr(\overline{A})$ is computed as the number of ways to place the $k$ balls inside the $n-1$ remaining bins, over all the possibilities to place the $k$ balls into the $n$ bins. 

\begin{equation*}
\begin{split}
Pr(A) &= 1 - Pr(\overline{A})\\
& = 1 - \frac{C_k^{n-1}}{C_k^n}\\
& = 1 - \frac{\frac{(n-1)!}{k! (n-1-k)!} }{\frac{n!}{k! (n-k)!}}\\
& = 1 - \frac{(n-k)!}{n(n-1-k)!}\\
& = 1 - \frac{(n-k)}{n}\\
& = \frac{k}{n}\\
\end{split}
\end{equation*} 
\end{proof}

Let X be the number of balls into a given bin after t rounds.
As the selection at each round are independent, the number of balls into a given bin follows
a binomial law : \\
$X \sim \operatorname{B} \left({t, p}\right)$ with $p = \frac{k}{n}$ (See Lemma~\ref{lemma3}) 
The expected value, denoted E(X), of the Binomial random variable X with parameters t and p is : 
$E(X) = tp = t\frac{k}{n}$

\end{document}